\documentclass[12pt]{article}
\usepackage{amsmath,amsthm,amssymb,amscd,fancybox,epsfig,float,subfigure}
\usepackage[margin=1in]{geometry}
\usepackage{graphics, hyperref}
\usepackage{epstopdf}
\usepackage{multirow}
\usepackage{hhline}
\usepackage{cleveref}
\usepackage{xcolor}
\usepackage{tikz}
\usetikzlibrary{arrows,automata}
\usepackage{natbib} 
\usepackage{pgfplots}

\newcommand{\beqn}{\begin{equation}}
\newcommand{\eeqn}{\end{equation}}

\newcommand{\bR}{\mathbb{R}}

\newcommand{\bbmat}{\begin{bmatrix}}
\newcommand{\ebmat}{\end{bmatrix}}

\usepackage[ruled,vlined,linesnumbered]{algorithm2e}
\newcommand{\R}{\mathbb{R}}

\def\bb{\begin{equation}}
\def\ee{\end{equation}}

\newtheorem{theorem}{Theorem}

\newtheorem{lemma}{Lemma}

\title{Fast methods for posterior inference of two-group normal-normal models}
\author{Philip Greengard\thanks{Department of Statistics, Columbia University,  Corresponding author, Email: pg2118@columbia.edu}, Jeremy Hoskins\thanks{Department of Statistics, University of Chicago}, Charles C.Margossian\thanks{Department of Statistics, Columbia University}, \\ Andrew Gelman\thanks{Department of Statistics and Political Science, Columbia University}, Aki Vehtari\thanks{Department of Computer Science, Aalto University}}
\date{2 Oct 2021}
\begin{document}

\maketitle

\begin{abstract}
We describe a class of algorithms for evaluating posterior moments of certain Bayesian linear regression models with a normal likelihood and a normal prior on the regression coefficients. The proposed methods can be used for hierarchical mixed effects models with partial pooling over one group of predictors, as well as random effects models with partial pooling over two groups of predictors. We demonstrate the performance of the methods on two applications, one involving U.S. opinion polls and one involving the modeling of COVID-19 outbreaks in Israel using survey data. 
The algorithms involve analytical marginalization of regression coefficients followed by numerical integration of the remaining low-dimensional density. The dominant cost of the algorithms is an eigendecomposition computed once for each value of the outside parameter of integration. Our approach drastically reduces run times compared to state-of-the-art Markov chain Monte Carlo (MCMC) algorithms. The latter, in addition to being computationally expensive, can also be difficult to tune when applied to hierarchical models.
\end{abstract}

\tableofcontents

\section{Introduction}
Advances over the last decade in statistical methods and their implementation 
in open-source, user-friendly software have drastically simplified statistical 
modeling for applied researchers. 
For example, with  probabilistic programming languages such as Stan \citep{carpenter1}
a user can specify and sample from a very general choice of posterior density 
with flexible language and an easy-to-use interface. For its primary tool of 
inference, Stan (as well as other probabilistic programming
languages) samples from the posterior distribution via 
dynamic Hamiltonian Monte Carlo sampler (HMC) \citep{betancourt:2018, hoffman:2014}. HMC is a grandient-based sampling method that 
has become ubiquitous in statistics over the last decade due to its being flexible, 
reliable, and general. 

Despite its widespread use, HMC, as well as other Markov chain Monte Carlo 
(MCMC) methods, have a substantial drawback in statistical 
problems with large amounts of data -- they can be prohibitively slow (and difficult to tune \citep[e.g.][]{betancourt:2015}).  For example, in the case of a linear regression
with $n$ observations and $k$ predictors, evaluation of the posterior density requires $O(nk)$ operations with straightforward implementation. To make matters worse, MCMC methods require large numbers of evaluations of the posterior density, and in the case of HMC, the posterior's gradient. 

Alternative methods for inference have been proposed for problems where 
MCMC is impractical. These approaches typically involve
a suitable approximation of the posterior density with a function with 
desirable properties. Laplace approximation methods \citep[e.g.][]{Margossian:2020} and variational inference \citep{vi} are two examples. More generally, there is 
extensive literature on efficient computational tools and analysis of posterior 
densities, and there are various software packages devoted to their implementation
\citep[see, e.g.][]{inla, tmb}. 

While these packages, and indeed most of the literature, are devoted to general tools for a wide range of posterior densities, in this paper we introduce an efficient algorithm for computing posterior 
expectations for two particular classes of Bayesian regression models---two-group 
normal-normal models and mixed-effects models. 
These classes of models find a broad range of applications in, for example, 
social sciences, epidemiology, biochemistry, and environmental sciences \citep{bda, gelman_hill_2006, greenland1, merlo1, 35}. 
Furthermore, in the broader context of model development, these regression models 
can serve as template models \citep{bayesian_workflow}.  

Using general MCMC methods for sampling from these posteriors can be exceedingly slow for problems with large amounts of data. By specializing to this particular family of models, we leverage their structure to create customized algorithms for fast and accurate inference.

The two Bayesian linear regression models we consider are:

\begin{enumerate}

\item \textbf{Two group normal-normal}: 
We define the two-group normal-normal model by
\bb\label{16}
\begin{split}
&y \sim \text{normal}(X_1\beta_1 + X_2\beta_2, \sigma_3) \\
&\beta_1 \sim \text{normal}(0, \sigma_1)\\
&\beta_2 \sim \text{normal}(0, \sigma_2),
\end{split}
\ee
where $X_1$ is a $n \times k_1$ matrix, $\beta_1 \in \R^{k_1}$ is a vector
of regression coefficients, $X_2$ is a $n \times 
k_2$ matrix,  and $\beta_2 \in \R^{k_2}$ is a vector of regression coefficients.
For Bayesian inference, we assume priors on the scale parameters 
$\sigma_1, \sigma_2, \sigma_3$. The performance of the algorithm is largely independent to the choice of these priors.
In the models that we use in this paper, we assign independent weakly informative $\text{normal}^+(0, 1)$ priors on the variance parameters $\sigma_1,\sigma_2, \sigma_3$ (assuming $y$ and the columns of $X$ have been normalized to have standard deviation $1$).

\item \textbf{Mixed effects}: 
The mixed-effects model differs slightly from the two-group normal-normal model. Instead of modeling the scale parameter $\sigma_2$, fixed 
scale parameters are assigned to the normal priors on $\beta_2$. The mixed-effects model is defined by
\bb\label{17}
\begin{split}
&y \sim \text{normal}(X_1\beta_1 + X_2\beta_2, \sigma_3) \\
&\beta_1 \sim \text{normal}(0, \sigma_1)\\
&\beta_{2,i} \sim \text{normal}(0, \sigma_{2,i}),
\end{split}
\ee
where $\sigma_{2,i}$ is the fixed scale parameter prior on each regression coefficient $\beta_{2,i}$ for $i=1,...,k_2$ where $\beta_2 \in \R^{k_2}$. We will assume priors on the scale parameters $\sigma_1, \sigma_3$.

\end{enumerate}

The models we discuss in this paper are standard models of Bayesian statistics and appear when seeking to model an outcome, $y$, as a linear combination of two (or more) distinct groups of predictors.
The Gaussian prior on the predictors enable various strategies commonly used in statistical modeling and machine learning; notably regularization and partial pooling between various sources of data. We demonstrate these models on three applications. 

\begin{enumerate}

\item
\textbf{COVID-19}: Due to a lack of reliable, fast, and widespread testing, an online survey initiative was created in Israel~\citep{Rossman:2020} for tracking and predicting COVID-19 outbreaks. We constructed a mixed-effects model for estimating geographic and age effects on the spread of the virus. With tens of thousands of responses, straightforward implementation of MCMC methods takes hours. Using the methods of this paper, we obtain accurate posterior inference in seconds. 

\item
\textbf{Rat growth}: We demonstrate the efficiency of our two-group algorithm on the classical two-group model for rat growth~\citep{gelfand:1990}, which estimates the growth rates of a population of rats over the first few weeks of life. 

\item
\textbf{Public opinion on abortion}: We use 2018 results of the annual Cooperative Congressional Election Study (CCES) to estimate geographic and demographic effects on attitudes towards abortion. The CCES contains nearly $100,000$ responses, and performing inference via MCMC sampling can be prohibitively slow. We use the mixed-effects algorithm introduced in this paper to perform posterior inference in seconds. 

\end{enumerate}

The computational methods we introduce for the two-group normal-normal
model and the mixed-effects models are closely related. In fact, the mixed-effects
model is a special case of the two-group model. We organize 
this paper by first describing our algorithm for the two-group normal-normal 
model in detail, and then outline the minor modifications that allow for efficient 
evaluation of mixed-effects models. 

The unnormalized density corresponding to the two-group model 
is given by
\bb \label{10}
q(\beta,\sigma_1,\sigma_2,\sigma_3) = \frac{e^{-\sigma_1^2/2-\sigma_2^2/2-\sigma_3^2/2}}{\sigma_1^n \sigma_2^{k_1}\sigma_3^{k_2}} e^{-\frac{1}{2\sigma_1^2} \|X \beta - y\|^2}e^{-\frac{1}{2\sigma_2^2}\|\beta_1\|^2}e^{-\frac{1}{2\sigma_3^2}\|\beta_2\|^2},
\ee
where $\beta = (\beta_1, \beta_2)$ with $\beta_1 \in \mathbb{R}^{k_1}, \beta_2 \in \mathbb{R}^{k_2}, \beta \in \mathbb{R}^k,$ and $y \in \mathbb{R}^n.$ For convenience, we
will be denoting by $\sigma$ the vector of scale parameters 
$(\sigma_1, \sigma_2, \sigma_3) \in \R^3$.


In the methods of this paper, we compute posterior moments of $q$ by 
analytically reducing the calculation of moments from integrals over $k+3$
dimensions to $3$-dimensional integrals. We then integrate 
the remaining $3$-dimensional integrals with a tensor product of 
Gaussian nodes. For example, we evaluate the normalizing constant
$C$ and posterior means for the regression coefficients, $\beta$, via
\begin{align*}
& C  = \int_{0}^{\infty} \int_{-\infty}^{\infty} \, q(\beta, \sigma_1,\sigma_2,\sigma_3) \, d\beta \, d\sigma 
\approx \sum_{i=1}^n f(\sigma_i) w_i \\
& E[\beta_j] = \int_{0}^{\infty} \int_{-\infty}^{\infty} \beta_j \, q(\beta, \sigma_1,\sigma_2, \sigma_3) \, d\beta \, d\sigma 
\approx \frac{1}{C} \sum_{i=1}^{n} f_j(\sigma_i) w_i ,
\end{align*}
where $\sigma_i \in \R^3$ and $w_i \in \R$ are three-dimensional 
Gaussian nodes and weights~\citep{trefethen} and 
\begin{align}
& f(\sigma_i) = \int_{\R^k} \, q(\beta, \sigma_i) d\beta \label{747} \\
& f_j(\sigma_i) = \int_{\R^k} \beta_j \, q(\beta, \sigma_i) d\beta \label{757}.
\end{align}
Integrals (\ref{747}) and (\ref{757}) can be evaluated
analytically via well-known equations~\citep{lindley1}, but
a straightforward implementation of those equations results
in a computational cost of $O(m^3 k^3)$ operations where $m$ is 
the number of discretization nodes needed in each dimension. In the methods
of this paper, we improve the computational cost of those integrals to 
$O(mk^3 + m^2k^2 + m^3)$ operations after a change of variables, allowing for the rapid evaluation of marginals. 


The tools used in the algorithm of this paper are a generalization of the 
approach proposed by \citet{greeng1} and generalize to higher-dimensional 
multilevel and higher-dimensional multigroup posterior distributions. Since we 
integrate the marginal density using a tensor product of Gaussian 
nodes, the cost of the integration scales like $O(m^d)$ where $m$ is the 
number of discretization nodes in each direction and $d$ is the 
dimension of the marginalized integral (where $d=3$ in the models of this paper). As a result, higher dimensional 
problems require evaluation of marginal integrals via sampling-based 
algorithms and cannot rely solely on Gaussian quadrature. We leave the analysis 
and description of numerical tools for such models to a subsequent publication.

The structure of this paper is as follows. In the following section we describe the change of variables of (\ref{10}) and provide mathematical analysis that will be used in the algorithm of this paper. Section \ref{s30} includes formulas that will allow for the evaluation of the normalizing constant of (\ref{10}). In Section \ref{s40} we describe analysis that will be used in computing expectations and in Section \ref{s50} we describe formulas for second moments. In Section \ref{s60} we discuss details of the implementation of the algorithm. Section \ref{s101} contains an algorithm for a special case of the two-group normal-normal model in which one group is much smaller than the other. The mixed effects algorithm, or rather the modification of the two-group normal-normal algorithm, is contained in Section \ref{s102}. In Section \ref{sec:covid}, Section \ref{sec:rats} and Section \ref{sec:abort} we apply the algorithms of this paper to applications. Conclusions and generalizations of the algorithm of  this paper are presented in Section \ref{sec:conclusion}.

\section{Mathematical apparatus}\label{s20}
Over the next several sections, we describe a numerical algorithm 
for computing expectations and second moments of the density 
$q: \R^{k+3} \rightarrow \R^+$ defined by 
$$
q(\beta,\sigma_1,\sigma_2,\sigma_3) = \frac{e^{-\sigma_1^2/2-\sigma_2^2/2-\sigma_3^2/2}}{\sigma_1^n \sigma_2^{k_1}\sigma_3^{k_2}} e^{-\frac{1}{2\sigma_1^2} \|X\beta - y\|^2}e^{-\frac{1}{2\sigma_2^2}\|\beta_1\|^2}e^{-\frac{1}{2\sigma_3^2}\|\beta_2\|^2},
$$
where $\sigma_1, \sigma_2, \sigma_3 > 0$, $\beta = (\beta_1, \beta_2)$ with $\beta_1 \in \mathbb{R}^{k_1}, \beta_2 \in \mathbb{R}^{k_2}, \beta \in \mathbb{R}^k,$ and $y \in \mathbb{R}^n.$
We begin by introducing notation that will be used throughout the numerical 
sections of this paper. 

Let $y = X\tilde{\beta} + d,$ where $X^td =0$, so that $\tilde{\beta}$ is the least-squares solution to the linear system $X\beta=y$ and $d$ is the residual. Let $I_1$ be the diagonal $k \times k$ matrix with ones in the first $k_1$ places on the diagonal, and zeroes in the remaining $k_2$ places. Similarly, let $I_2$ be the diagonal $k \times k$ matrix with zeroes in the first $k_1$ places on the diagonal, and one in the remaining $k_2$ places.

We perform a change of variables in $\sigma_1,\sigma_2,$ and $\sigma_3,$ defining $\rho, \theta,$ and $\phi$ implicitly by
\begin{align*}
\sigma_1 &= \rho \cos \phi,\\
\sigma_2 &= \rho \sin{\phi} \cos{\theta},\\
\sigma_3 &= \rho \sin{\phi} \sin{\theta}.
\end{align*}
This corresponds to changing to spherical coordinates in the $\sigma$ variables. With these substitutions, and with some minor abuse of notation, we obtain
\begin{align*}
f(\beta,\rho, \theta,\phi) = &\frac{e^{-\rho^2/2}}{\rho^{n+k} \cos^n(\phi)\sin^k{\phi} \,\cos^{k_1}{\theta} \,\sin^{k_2}{\theta}}\nonumber\\
& \exp\left[-\frac{1}{2\rho^2}\left(\frac{1}{\cos^2\phi} \|X(\beta-\tilde{\beta})\|^2 + \frac{\|d\|^2}{\cos^2{\phi}} +\frac{\beta^t \left(\frac{I_1}{\cos^2\theta} + \frac{I_2}{\sin^2\theta} \right)\beta}{\sin^2{\phi}} \right)\right].
\end{align*}
The differentials transform as follows:
\begin{align*}
{\rm d}\sigma_1\,{\rm d}\sigma_2\,{\rm d}\sigma_3 = \rho^2 \sin{\phi} \, {\rm d}\rho\,{\rm d}\theta\,{\rm d}\phi.
\end{align*}
Moreover, the condition that $\sigma_1,\, \sigma_2, \sigma_3 >0,$ is equivalent to $0 < \phi,\theta < \pi/2.$

\section{Normalizing constant}\label{s30}
In this section we describe the computation of the integral of $f$ over its domain. If we denote this quantity by $I_0,$ then $f/I_0$ is a probability density on $\mathbb{R}^+\times\mathbb{R}^+\times\mathbb{R}^+\times \mathbb{R}^k.$ We begin by making a change of variables in $\beta,$ setting $\beta = ( I_1 \cos{\theta} +I_2 \sin{\theta})z.$ Similarly, we define $\tilde{z}$ implicitly by $\tilde{\beta} = ( I_1\cos{\theta} +I_2 \sin{\theta})\tilde{z}.$ Then,
\begin{align*}
f(z,\rho, \theta,\phi) = &\frac{e^{-\rho^2/2}}{\rho^{n+k} \sin^k{\phi} \,\cos^n{\phi}\,\cos^{k_1}{\theta} \,\sin^{k_2}{\theta}}\nonumber\\
& \exp\left[-\frac{1}{2\rho^2}\left(\frac{1}{\cos^2\phi} \|X_\theta (z-\tilde{z})\|^2 + \frac{\|d\|^2}{\cos^2{\phi}} +\frac{\|z\|^2}{\sin^2{\phi}} \right)\right],
\end{align*}
where $X_\theta = X \,( I_1 \cos{\theta} +I_2  \sin{\theta}).$ The differentials transform as follows
\begin{align*}
{\rm d}\beta_1\, \dots \,{\rm d}\beta_k =\cos^{k_1}\theta\,\sin^{k_2}\theta\,\, {\rm d}z_1\, \dots \,{\rm d}z_k. 
\end{align*}

To diagonalize the quadratic form appearing in the exponent, we perform an eigendecomposition of on $X_\theta^t X_{\theta}$ obtaining
\begin{align}\label{340b}
X_\theta^t X_\theta=  V_\theta D_\theta V^t_\theta,
\end{align}
where $D_\theta \in \mathbb{R}^{k \times k}$ is diagonal and positive, and $V_\theta \in \mathbb{R}^{k \times k}$ is a unitary matrix. We have assumed here that $n \ge k.$ If the converse is true then a small modification is required. In the following, for notational convenience we denote the diagonal entries of $D_\theta$ by $\lambda_i(\theta).$

Next we set $z = V_\theta w$ and $\tilde{z} = V_\theta \tilde{w}.$ In terms of the original variables, $\beta = (I_1 \cos \theta + I_2 \sin \theta)V_\theta w$ and $\tilde{\beta}  = (I_1 \cos \theta + I_2 \sin \theta)V_\theta \tilde{w}.$ In particular,
\begin{align*}
{\rm d}\beta_1\, \dots \,{\rm d}\beta_k =\cos^{k_1}\theta\,\sin^{k_2}\theta\,\, {\rm d}w_1\, \dots \,{\rm d}w_k. 
\end{align*}

After making these substitutions, we obtain
\begin{align*}
\int &\dots \int f(\beta, \sigma_1,\sigma_2,\sigma_3)\,{\rm d}\beta_1\, \dots \,{\rm d}\beta_k \nonumber\\
&= \frac{e^{-\frac{\rho^2}{2}-\frac{\|d\|^2}{2 \rho^2 \cos^2{\phi}}}}{\rho^{n+k}\cos^n\phi \,\sin^k\phi} \int \dots \int \,\exp\left[-\frac{1}{2\rho^2}\sum_{i=1}^k  \left(\frac{\lambda_i (w_i-\tilde{w}_i)^2}{\cos^2\phi}+ \frac{w_i^2}{\sin^2\phi} \right)\right]\,{\rm d}w_1\, \dots \,{\rm d}w_k \\
&= \frac{e^{-\frac{\rho^2}{2}-\frac{\|d\|^2}{2 \rho^2 \cos^2{\phi}}}}{\rho^{n+k}\cos^n\phi \,\sin^k\phi} \prod_{i=1}^k \int \,\exp\left[-\frac{1}{2\rho^2}  \left(\frac{\lambda_i (w_i-\tilde{w}_i)^2}{\cos^2\phi}+ \frac{w_i^2}{\sin^2\phi} \right)\right]\,{\rm d}w_i.
\end{align*}
Thus, the integrals over the $\beta$ variables have been reduced to the product of $k$ one-dimensional Gaussian integrals. Using the identity
\begin{align*}
\int_\mathbb{R} e^{-\frac{a}{2} (s-s_0)^2 -\frac{b}{2} s^2} {\rm d}s = \sqrt{\frac{2\pi}{a+b}}e^{-\frac{ab \,s_0^2}{2(a+b)}},
\end{align*}
we find that
\begin{align*}
\int &\dots \int f(\beta, \sigma_1,\sigma_2,\sigma_3)\,{\rm d}\beta_1\, \dots \,{\rm d}\beta_k \\
&= \frac{e^{-\frac{\rho^2}{2}-\frac{\|d\|^2}{2 \rho^2 \cos^2{\phi}}}}{\rho^{n+k}\cos^n\phi \,\sin^k\phi} \prod_{i=1}^k \sqrt{\frac{2\pi \rho^2}{\frac{\lambda_i}{\cos^2\phi}+\frac{1}{\sin^2\phi}}} \exp\left[-\frac{1}{2\rho^2} \tilde{w_i}^2 \frac{\lambda_i}{\cos^2\phi \sin^2 \phi \left(\frac{\lambda_i}{\cos^2\phi}+\frac{1}{\sin^2\phi} \right)} \right]\\
&= \frac{e^{-\frac{\rho^2}{2}-\frac{\|d\|^2}{2 \rho^2 \cos^2{\phi}}}}{\rho^{n}\cos^{n-k}\phi } \prod_{i=1}^k \sqrt{\frac{2\pi }{\lambda_i \sin^2\phi +\cos^2 \phi}} \exp\left[-\frac{1}{2\rho^2}  \frac{\lambda_i \tilde{w_i}^2}{{\lambda_i \sin^2\phi}+{\cos^2\phi} } \right].
\end{align*}
Next, we define the functions $\alpha: (0,\pi/2)^2 \to \mathbb{R}^+$ and $\beta: (0,\pi/2)^2 \to \mathbb{R}^+$ by
\begin{align*}
\alpha(\phi,\theta) = \prod_{i=1}^k \sqrt{\frac{2\pi }{\lambda_i(\theta) \sin^2\phi +\cos^2 \phi}} ,\\
\beta(\phi,\theta) = \sum_{i=1}^k \frac{\lambda_i(\theta)\tilde{w}_i^2(\theta)}{\lambda_i(\theta) \sin^2 \phi + \cos^2\phi}.
\end{align*}
Then
\begin{align*}
\int &\dots \int f(\beta, \sigma_1,\sigma_2,\sigma_3)\,{\rm d}\beta_1\, \dots \,{\rm d}\beta_k=\frac{e^{-\frac{\rho^2}{2}-\frac{\|d\|^2}{2 \rho^2 \cos^2{\phi}}}}{\rho^{n}\cos^{n-k}\phi } \alpha(\phi,\theta) e^{-\frac{1}{2\rho^2} \beta(\phi,\theta)}.
\end{align*}
For a fixed $\theta,$ the vector $\tilde{w}$ and the eigenvalues values $\lambda_i$ are independent of $\rho$ and $\phi,$ and hence need to be recomputed only when $\theta$ is changed. Moreover, for fixed $\theta$ and $\phi,$ the above integral can be computed in $O(1)$ floating operations for each new value of $\rho.$ In particular, the normalization constant can be computed efficiently via the formula
\begin{align}\label{360}
I_0 = \int_0^{\frac{\pi}{2}}\int_0^{\frac{\pi}{2}}\frac{\alpha(\phi,\theta)  \sin\phi}{\cos^{n-k}\phi}\int_0^\infty\frac{e^{-\frac{\rho^2}{2}-\frac{1}{2\rho^2}\left(\frac{\|d\|^2}{\cos^2\phi} + \beta(\phi,\theta)\right)}}{\rho^{n}}\rho^2 \,{\rm d}\rho\, {\rm d}\phi\,{\rm d}\theta.
\end{align}

\section{Expectations of \texorpdfstring{$\beta$}{β}}\label{s40}
In this section we describe how to compute moments in $\beta$ of the distribution $f$. We begin by observing that
\begin{align*}
\beta_\ell = Q_\ell(\theta) \sum_{\ell=1}^k \left(V_\theta\right)_{\ell,j} {w}_j,
\end{align*}
where $Q_\ell(\theta) = \cos \theta$ if $1 \le \ell \le k_1,$ and $\sin \theta$ if $k_1< \ell \le k.$ Let $M_j (\rho,\phi,\theta)$ be defined by
\begin{align*}
M_j(\rho_\phi,\theta)=&\frac{e^{-\frac{\rho^2}{2}-\frac{\|d\|^2}{2 \rho^2 \cos^2{\phi}}}}{\rho^{n+k}\cos^n\phi \,\sin^k\phi} \prod_{i\neq j}^k \int \,\exp\left[-\frac{1}{2\rho^2}  \left(\frac{\lambda_i (w_i-\tilde{w}_i)^2}{\cos^2\phi}+ \frac{w_i^2}{\sin^2\phi} \right)\right]\,{\rm d}w_i\\
&\quad \times \int \,w_j\,\exp\left[-\frac{1}{2\rho^2}  \left(\frac{\lambda_j (w_j-\tilde{w}_j)^2}{\cos^2\phi}+ \frac{w_j^2}{\sin^2\phi} \right)\right]\,{\rm d}w_j.
\end{align*}
Using the results of the previous section we can perform the integrals to obtain
\begin{align*}
M_j(\rho,\phi,\theta) =& \frac{e^{-\frac{\rho^2}{2}-\frac{\|d\|^2}{2 \rho^2 \cos^2{\phi}}}}{\rho^{n+k}\cos^n\phi \,\sin^k\phi} \prod_{i\neq j}^k \sqrt{\frac{2\pi \rho^2}{\frac{\lambda_i}{\cos^2\phi}+\frac{1}{\sin^2\phi}}} \exp\left[-\frac{1}{2\rho^2} \tilde{w_i}^2 \frac{\lambda_i}{\cos^2\phi \sin^2 \phi \left(\frac{\lambda_i}{\cos^2\phi}+\frac{1}{\sin^2\phi} \right)} \right]\\
&\quad \times \frac{\lambda_j\tilde{w}_j}{\cos^2{\phi}\left(\frac{\lambda_j}{\cos^2\phi}+\frac{1}{\sin^2\phi} \right)} \sqrt{\frac{2\pi \rho^2}{\frac{\lambda_j}{\cos^2\phi}+\frac{1}{\sin^2\phi}}} \exp\left[-\frac{1}{2\rho^2} \tilde{w_j}^2 \frac{\lambda_j}{\cos^2\phi \sin^2 \phi \left(\frac{\lambda_j}{\cos^2\phi}+\frac{1}{\sin^2\phi} \right)} \right]\\
=&\left(\frac{\lambda_j \tilde{w}_j}{\lambda_j \sin^2\phi+ \cos^2\phi} \right) \sin^2 \phi\left[\frac{e^{-\frac{\rho^2}{2}-\frac{\|d\|^2}{2 \rho^2 \cos^2{\phi}}}}{\rho^{n}\cos^{n-k}\phi } \alpha(\phi,\theta) e^{-\frac{1}{2\rho^2} \beta(\phi,\theta)}\right].
\end{align*}
We remark that the second factor in the above expression is the same as the one arising in the computation of $I_0.$ The first factor depends on $j,$ $\phi,$ and $\theta$ but not on $\rho.$ For ease of exposition, let us define $\tilde{M}_j$ by
\begin{align}\label{817}
\tilde{M}_j (\phi,\theta) = \left(\frac{\lambda_j \tilde{w}_j}{\lambda_j \sin^2\phi+ \cos^2\phi} \right).
\end{align}
Then
\begin{align}\label{120}
E[\beta_\ell] = \sum_{\ell,j} \frac{1}{I_0}\int_0^{\frac{\pi}{2}}Q_\ell(\theta)\,(V_\theta)_{\ell,j}\int_0^{\frac{\pi}{2}}\frac{\alpha(\phi,\theta)  \sin^3\phi}{\cos^{n-k}\phi}\tilde{M}_j(\phi,\theta)\,\int_0^\infty\frac{e^{-\frac{\rho^2}{2}-\frac{1}{2\rho^2}\left(\frac{\|d\|^2}{\cos^2\phi} + \beta(\phi,\theta)\right)}}{\rho^{n}}\rho^2 \,{\rm d}\rho\, {\rm d}\phi\,{\rm d}\theta.
\end{align}
In particular, if $I_1^{(j)}$ is defined via the formula
\begin{align}\label{130}
I_1^{(j)}(\theta) =\int_0^{\frac{\pi}{2}}\frac{\alpha(\phi,\theta)  \sin^3\phi}{\cos^{n-k}\phi}\tilde{M}_j(\phi,\theta)\,\int_0^\infty\frac{e^{-\frac{\rho^2}{2}-\frac{1}{2\rho^2}\left(\frac{\|d\|^2}{\cos^2\phi} + \beta(\phi,\theta)\right)}}{\rho^{n}}\rho^2 \,{\rm d}\rho\, {\rm d}\phi,
\end{align}
then 
\begin{align}\label{345}
E[\beta_\ell] =\int_0^{\frac{\pi}{2}} \sum_{\ell,j} Q_i(\theta) (V_\theta)_{\ell,j} I_1^{(j)}(\theta)\,{\rm d}\theta,
\end{align}
and hence all $k$ moments can be computed simultaneously with an integrand requiring $O(k^2)$ operations to compute (assuming the number of quadrature nodes required to achieve a fixed precision is more or less independent of $k$).

\section{Covariance of \texorpdfstring{$\beta$}{β}}\label{s50}
In this section, we describe formulas for computing the posterior covariance matrix of $\beta$. We use the identity 
\bb\label{110}
-\frac{a}{2}(s-s_0)^2 - \frac{b}{2}s^2 = -\frac{a+b}{2}\Bigg(s - \frac{as_0}{-(a+b)}\Bigg)^2 + \frac{-abs_0^2}{2(a+b)}
\ee
to compute second moments of $w_j$. That is, 
letting $P_j (\rho,\phi,\theta)$ be defined by
\begin{align*}
P_j(\rho, \phi,\theta)=&\frac{e^{-\frac{\rho^2}{2}-\frac{\|d\|^2}{2 \rho^2 \cos^2{\phi}}}}{\rho^{n+k}\cos^n\phi \,\sin^k\phi} \prod_{i\neq j}^k \int \,\exp\left[-\frac{1}{2\rho^2}  \left(\frac{\lambda_i (w_i-\tilde{w}_i)^2}{\cos^2\phi}+ \frac{w_i^2}{\sin^2\phi} \right)\right]\,{\rm d}w_i\\
&\quad \times \int \,w_j^2\,\exp\left[-\frac{1}{2\rho^2}  \left(\frac{\lambda_j (w_j-\tilde{w}_j)^2}{\cos^2\phi}+ \frac{w_j^2}{\sin^2\phi} \right)\right]\,{\rm d}w_j \\
=&\frac{e^{-\frac{\rho^2}{2}-\frac{\|d\|^2}{2 \rho^2 \cos^2{\phi}}}}{\rho^{n+k}\cos^n\phi \,\sin^k\phi} \prod_{i\neq j}^k \int \,\exp\left[-\frac{1}{2\rho^2}  \left(\frac{\lambda_i (w_i-\tilde{w}_i)^2}{\cos^2\phi}+ \frac{w_i^2}{\sin^2\phi} \right)\right]\,{\rm d}w_i\\
&\quad \times \int \,w_j^2\,\exp\left[-\frac{1}{2\rho^2}  \left(\frac{\lambda_j (w_j-\tilde{w}_j)^2}{\cos^2\phi}+ \frac{w_j^2}{\sin^2\phi} \right)\right]\,{\rm d}w_j
.
\end{align*}
we use (\ref{110}) to obtain
\begin{align*}
P_j(\rho,\phi,\theta) =& \frac{e^{-\frac{\rho^2}{2}-\frac{\|d\|^2}{2 \rho^2 \cos^2{\phi}}}}{\rho^{n+k}\cos^n\phi \,\sin^k\phi} \prod_{i\neq j}^k \sqrt{\frac{2\pi \rho^2}{\frac{\lambda_i}{\cos^2\phi}+\frac{1}{\sin^2\phi}}} \exp\left[-\frac{1}{2\rho^2} \tilde{w_i}^2 \frac{\lambda_i}{\cos^2\phi \sin^2 \phi \left(\frac{\lambda_i}{\cos^2\phi}+\frac{1}{\sin^2\phi} \right)} \right]\\
&\quad \times \frac{\lambda_j\tilde{w}_j}{\cos^2{\phi}\left(\frac{\lambda_j}{\cos^2\phi}+\frac{1}{\sin^2\phi} \right)} 
\sqrt{\frac{2\pi \rho^2}{\frac{\lambda_j}{\cos^2\phi}+\frac{1}{\sin^2\phi}}} \exp\left[-\frac{1}{2\rho^2} \tilde{w_j}^2 \frac{\lambda_j}{\cos^2\phi \sin^2 \phi \left(\frac{\lambda_j}{\cos^2\phi}+\frac{1}{\sin^2\phi} \right)} \right]\\
=&\left(\frac{\rho^2}{\frac{\lambda_j}{\cos^2\phi} + \frac{1}{\sin^2\phi}} \right) \sin^2 \phi\left[\frac{e^{-\frac{\rho^2}{2}-\frac{\|d\|^2}{2 \rho^2 \cos^2{\phi}}}}{\rho^{n}\cos^{n-k}\phi } \alpha(\phi,\theta) e^{-\frac{1}{2\rho^2} \beta(\phi,\theta)}\right]
- M_j(\rho, \phi,\theta)^2.
\end{align*}
Defining $\tilde{P}_j$ by
\begin{align*}
\tilde{P}_j (\phi,\theta) = \frac{\rho^2}{\frac{\lambda_j}{\cos^2\phi} + \frac{1}{\sin^2\phi}}\end{align*}
and defining $I_2^{(j)}$ via the formula
\begin{align}\label{130b}
I_2^{(j)}(\theta) =\int_0^{\frac{\pi}{2}}\frac{\alpha(\phi,\theta)  \sin^3\phi}{\cos^{n-k}\phi}\tilde{P}_j(\phi,\theta)\,\int_0^\infty\frac{e^{-\frac{\rho^2}{2}-\frac{1}{2\rho^2}\left(\frac{\|d\|^2}{\cos^2\phi} + \beta(\phi,\theta)\right)}}{\rho^{n}}\rho^2 \,{\rm d}\rho\, {\rm d}\phi,
\end{align}
we observe that
\begin{align*}
E[\beta\beta^t]_{i,m} = \int_0^{\frac{\pi}{2}} (C I_2(\theta) C^t)_{i,j}
\,d\theta
\end{align*}
where $C$ is the $k \times k$ matrix defined by
\begin{align*}
C_{i,j} = \sum_{m} Q_i(\theta) (V_\theta)_{m,j}.
\end{align*}
We then compute the posterior covariance of $\beta$ via
\bb
\begin{split}
E[(\beta-E[\beta])(\beta-E[\beta])^t] = E[\beta\beta^t] - E[\beta]E[\beta]^t,
\end{split}
\ee
where $E[\beta]$ is obtained via (\ref{345}).

\section{Numerical implementation}\label{s60}
We now describe a numerical approach for computing the normalizing constant $I_0$, see (\ref{360}), and moments of $q$ (see (\ref{10})). The quadrature rules used provide arbitrary user-specified precision for both the normalizing constant as well as moments. Our integration scheme is a tensor product of Gaussian nodes in the $\theta$, $\phi$, and $\rho$ directions. 

In the remainder of this section we describe how to adaptively determine integration bounds in the $\phi$ and $\rho$ directions as well as the number of nodes to use in $\theta$. 

\subsection*{Integrals with respect to \texorpdfstring{$\rho$}{ρ}}
We first describe an approach for integrating the inner integral,
\bb\label{310}
\int_0^\infty\frac{e^{-\frac{\rho^2}{2}-\frac{1}{2\rho^2}\left(\frac{\|d\|^2}{\cos^2\phi} + \beta(\phi,\theta)\right)}}{\rho^{n-2}}\,{\rm d}\rho.
\ee
We do this by finding upper and lower integration bounds and then performing Gaussian quadrature with $80$ nodes. 
The maximum of the integrand of (\ref{310}), which we denote 
$\rho_{max}$, is achieved at 
\begin{equation}
\rho_{max} = \frac{1}{\sqrt{2}}\left(\sqrt{4\frac{\|d\|^2}{\cos^2\phi} + 4\beta(\phi,\theta)+(n-2)^2}-n+2\right)^{1/2}.
\end{equation}
See Appendix \ref{app1} for details. Furthermore, for all $\rho > \rho_{max}$, the integrand is monotonically decreasing and for all $\rho < \rho_{max}$, the integrand is monotonically increasing. 
We choose our upper integration bound, $\rho_1$, to be the value $\rho_1 > \rho$ such that the integrand of (\ref{310}) is smaller than its maximum by a factor of $10^{20}$. 
We evaluate $\rho_1$ with bisection with initial bounds of
\bb
\rho_{max} \qquad \text{and} \qquad \rho_{max} + \frac{40}{\sqrt{-\sigma}},
\ee
where $\sigma$ is the second derivative of the integrand of (\ref{310}) with respect to $\rho$ evaluated at $\rho_{max}$ (see Appendix \ref{app1}). 
We also find the lower bound of integration using bisection where our 
starting bounds for the bisection are
0 and $\rho_{max}$.
After finding the bounds of integration, we evaluate (\ref{310}) using Gaussian quadrature. Using $80$ nodes was sufficient for full double precision accuracy in our examples. 

\subsection*{Integral with respect to \texorpdfstring{$\phi$}{φ}}
We evaluate the integral $I_1^{(j)}(\theta)$ (see (\ref{130})) using Gaussian quadrature where we determine the upper bound of integration adaptively. That is, we determine the upper bound of integration, $\phi_1$, by sequentially evaluating the integrand of (\ref{130}) at order $100$ Gaussian nodes until the integrand is smaller than the maximum observed value of the integrand by a factor of $10^{20}$. We then declare the first such point, $\phi_1$, to be the upper integration bound and evaluate (\ref{130}) with $80$ Gaussian nodes on the interval $(0, \phi_1)$.

\subsection*{Integral with respect to \texorpdfstring{$\theta$}{θ}}
In this section, we describe a technique for evaluating the outer integral of (\ref{345}). For each evaluation of the integrand of (\ref{345}), we perform eigendecomposition (\ref{340b}). As a result, each evaluation of the integrand requires $O(k^3)$ operations and the total computational time for evaluating moments of posterior (\ref{10}) is roughly linear in the number of evaluations of the integrand of (\ref{345}). 

We compute integral (\ref{345}) using Gaussian quadrature. 
We found that for a general set of problems, the integrand of 
(\ref{345}) was sufficiently smooth, that using around $10$ nodes 
was sufficient for several digits of accuracy. To check accuracy
of an $m$-point Gaussian quadrature we compute the same integral
with $2m$ nodes and check the difference, which is a 
proxy for the accuracy of the $m$-point Gaussian quadrature. 

For problems in which the number of evaluations of the integrand
of (\ref{345}) needs to be reduced, we can use the following approach. 
Approximate the integral using two Chebyshev nodes 
(see, e.g., \cite{trefethen}) and four Chebyshev nodes and compute the difference between the two values. That difference provides an estimate for the error of the approximation using two nodes. If the difference between the two approximations is larger than desired, then double the number of nodes and again compute the difference. Continue to double the number of nodes until the desired accuracy is achieved. 

By using practical Chebyshev nodes in $\theta$ the order $n$ nodes are a subset of the order $2n$ nodes. This saves $n$ evaluations of the integrand when approximating the integral with $2n$ nodes.

\begin{algorithm}[H]\label{a10}
\DontPrintSemicolon
\SetAlgoLined
Construct $\theta_1,...,\theta_m$, Gaussian nodes and weights in $\theta$ on $[0, 2\pi]$\;
\hspace{2em} For each $\theta_i$:\;
\hspace{4em} Compute eigendecomposition (\ref{340b})\;
\hspace{4em} Evaluate upper integration bound in $\phi$\;
\hspace{4em} Construct $\phi_1,...,\phi_m$, Gaussian nodes in $\phi$\;
\hspace{4em} For each $\phi_{\ell}$:\;
\hspace{6em} Evaluate integration bounds for $\rho$ integral\;
\hspace{6em} Compute integral using Gaussian nodes\;
\hspace{6em} Evaluate $\tilde{M}_j(\phi_{\ell}, \theta_i)$ of (\ref{817}) for $j=1,...,k$\;
\hspace{4em} Convert expectations back to $\beta$\;
 \caption{\em Two-group normal-normal models: Evaluation of posterior expectations.}
\end{algorithm}

\section{Special case of two-group model} \label{s101}
In this section we consider the special case where the posterior density $f$ corresponds to an intercept model and one of the two groups is much smaller than the other.  That is, $k_1$ is small relative to $k_2$, and each row of $X$ contains two non-zero entries---one in the first $k_1$ columns and another in the next $k_2$ columns. This model appears in applications in which we seek to model an outcome, $y$, as a combination of two unrelated factors. 

For the corresponding posterior, we introduce a fast algorithm that analytically marginalizes the normal-normal parameters $\beta$ by evaluating a determinant and solving a linear system.

The terms dependent on $\beta$ of the exponential of $f$ can be written as 
\begin{align*}
\frac{1}{\sigma_1^2} (\beta-\beta_0)^t X^t X (\beta-\beta_0) + \beta^t R \beta = (\beta-\bar{\beta}(\sigma))^t \left(\frac{1}{\sigma_1^2}X^tX +R\right) (\beta-\bar{\beta}(\sigma)) + C(\sigma),
\end{align*}
where 
\begin{align}\label{340}
\sigma_1^2\left(\frac{1}{\sigma_1^2}X^tX +R\right)\bar{\beta}(\sigma) = X^tX \beta_0= X^t y
\end{align}
and
\begin{align*}
C(\sigma) &= \frac{1}{\sigma_1^2}(\beta_0 - \bar{\beta}(\sigma))^t X^t X \beta_0 \\
&= \frac{1}{\sigma_1^2}(\beta_0-\bar{\beta}(\sigma))^t X^t y\\
&= \frac{1}{\sigma_1^2}(y-y_r-X \bar{\beta}(\sigma))^t  y\\
&= \frac{1}{\sigma_1^2}(y_p-X \bar{\beta}(\sigma))^t  y.
\end{align*}

Before proceeding further, we introduce a change of variables. Let 
$$(\sigma_1,\sigma_2,\sigma_3) = \rho (\cos\theta,\nu_2,\sin\theta)$$
and note that
$${\rm d} \sigma_1\,{\rm d} \sigma_2 \,{\rm d}\sigma_3 = \rho^2 {\rm d}\rho\,{\rm d} \nu_2 \,{\rm d}\theta.$$
Then $R$ of (\ref{340}) satisfies
\begin{align*}
R = \frac{1}{\rho^2}\begin{bmatrix} \frac{1}{\nu_2^2} I_{k_1} & 0 \\
0 & \frac{1}{\sin^2\theta} I_{k_2}
\end{bmatrix}.
\end{align*}
Let $\tilde{R}= \rho^2\cos^2{\theta} \, R$; this is a function only of $\nu_2$ and $\theta.$ It follows that
\begin{align*}
f(\beta,\rho,\nu_2,\theta) = \frac{1}{\rho^{n+k} \nu_2^{k_1} \cos^n\theta\, \sin^{k_2}\theta} e^{- \frac{\rho^2}{2} \left(1+ \nu_2^2 \right)- \frac{1}{2\rho^2\cos^2\theta}(\beta-\bar{\beta})^t \left(X^tX +\tilde{R}\right) (\beta-\bar{\beta}) - \frac{1}{2\rho^2 \cos^2\theta}\left[y_p^ty -\bar{\beta} X^t y +\|y_r\|^2\right]}.
\end{align*}
For ease of exposition, let $\beta$ be the function defined by
$$\beta(\nu_2,\theta) = \frac{1}{\cos^2\theta} \left[y_p^ty -\bar{\beta} X^t y +\|y_r\|^2\right].$$
The following lemma will be used in Theorem \ref{331} and will provide formulas for computing posterior moments of $f$. 
\begin{lemma}[moments of a Gaussian]
Let $C$ be a symmetric positive definite $k \times k$ matrix. Then, for any vector $\tilde{\beta} \in \mathbb{R}^k,$
\begin{enumerate}
\item
\begin{align*}
\int_{\mathbb{R}^k} e^{-(\beta-\tilde{\beta}) C (\beta-\tilde{\beta})} \,{\rm d}\beta = \frac{\pi^\frac{k}{2}}{\sqrt{\det{C}}}
\end{align*}
\item
\begin{align*}
\int_{\mathbb{R}^k} \beta\, e^{-(\beta-\tilde{\beta}) C (\beta-\tilde{\beta})} \,{\rm d}\beta = \frac{\pi^\frac{k}{2}}{\sqrt{\det{C}}} \,\tilde{\beta}
\end{align*}
\item
\begin{align*}
\int_{\mathbb{R}^k} \beta \beta^t\, e^{-(\beta-\tilde{\beta}) C (\beta-\tilde{\beta})} \,{\rm d}\beta = \frac{\pi^\frac{k}{2}}{\sqrt{\det{C}}} (\frac{1}{2}C^{-1}\,+\,\tilde{\beta}\tilde{\beta}^t)
\end{align*}
\end{enumerate}
\end{lemma}
\begin{proof}
The proof of the first two is immediate. For the second, let $z = \sqrt{C} (\beta-\tilde{\beta}).$ Then, the integral becomes
\begin{align*}
\int_{\mathbb{R}^k}& \left[\sqrt{C}^{-1}z +\tilde{\beta} \right]\,\left[\sqrt{C}^{-1}z +\tilde{\beta} \right]^t e^{-z^2}\, \frac{1}{\det{\sqrt{C}}} \,{\rm d}z\\
&=\int_{\mathbb{R}^k} \left(\sqrt{C}^{-1} z z^t \sqrt{C}^{-1} + \tilde{\beta} \tilde{\beta}^t \right) e^{-z^2}\, \frac{1}{\det{\sqrt{C}}} \,{\rm d}z\\
&= \frac{\pi^\frac{k}{2}}{\sqrt{\det{C}}} \left(\frac{1}{2} C^{-1} +\tilde{\beta} \tilde{\beta}^t \right).
\end{align*}
\end{proof}
The following theorem follows immediately from the previous lemma and provides formulas that will be used to compute posterior moments of $f$.
\begin{theorem}\label{331}
Let $f$ be the unnormalized probability density defined above. Then
\begin{enumerate}
\item
\begin{align*}
f_0 &:= \int_{\mathbb{R}^+}\int_{\mathbb{R}^+}\int_{\mathbb{R}^+}\int_{\mathbb{R}^k} f(\beta,\sigma)\,{\rm d}\beta\,{\rm d}\sigma\\
&= (2\pi)^\frac{k}{2}\int_{\mathbb{R}^+}\frac{1}{\cos^{n-k}\theta\,\sin^{k_2}\theta}\, \int_{\mathbb{R}^+}\frac{1}{\nu_2^{k_1}}\frac{1}{\sqrt{\det{X^tX+\tilde{R}}}}\int_{\mathbb{R}^+} \frac{e^{-\frac{\rho^2(1+\nu_2^2)}{2}- \frac{1}{2\rho^2}\beta}}{\rho^{n-2}}\,{\rm d} \rho \,{\rm d}\nu_2\,{\rm d}\theta
\end{align*}
\item
\begin{align*}
f_j &:= \int_{\mathbb{R}^+}\int_{\mathbb{R}^+}\int_{\mathbb{R}^+}\int_{\mathbb{R}^k} \beta_jf(\beta,\sigma)\,{\rm d}\beta\,{\rm d}\sigma\\
&= (2\pi)^\frac{k}{2}\int_{\mathbb{R}^+}\frac{1}{\cos^{n-k}\theta\,\sin^{k_2}\theta}\, \int_{\mathbb{R}^+}\frac{1}{\nu_2^{k_1}}\frac{\bar{\beta}_j(\nu_2,\theta)}{\sqrt{\det{X^tX+\tilde{R}}}}\int_{\mathbb{R}^+} \frac{e^{-\frac{\rho^2(1+\nu_2^2)}{2}- \frac{1}{2\rho^2}\beta}}{\rho^{n-2}}\,{\rm d} \rho \,{\rm d}\nu_2\,{\rm d}\theta
\end{align*}
\item
\begin{align*}
f_{jk} &:= \int_{\mathbb{R}^+}\int_{\mathbb{R}^+}\int_{\mathbb{R}^+}\int_{\mathbb{R}^k} \beta_j\beta_kf(\beta,\sigma)\,{\rm d}\beta\,{\rm d}\sigma\\
&= (2\pi)^\frac{k}{2}\int_{\mathbb{R}^+}\frac{1}{\cos^{n-k}\theta\,\sin^{k_2}\theta}\, \int_{\mathbb{R}^+}\frac{1}{\nu_2^{k_1}}\frac{\frac{1}{2}(X^tX+\tilde{R})^{-1}_{jk} +\bar{\beta}_j\bar{\beta}_k }{\sqrt{\det{X^tX+\tilde{R}}}}\int_{\mathbb{R}^+} \frac{e^{-\frac{\rho^2(1+\nu_2^2)}{2}- \frac{1}{2\rho^2}\beta}}{\rho^{n-2}}\,{\rm d} \rho \,{\rm d}\nu_2\,{\rm d}\theta.
\end{align*}
\end{enumerate}
\end{theorem}
\subsection{Numerical apparatus}
In this section, we describe a numerical method for exploiting the structure of $X^tX$ and $\tilde{R}$ to obtain a computationally efficient strategy for computing moments of $f.$ 

We use the following lemma to compute the determinant found in each integral of Theorem \ref{331}. The lemma follows immediately from the Schur complement formula  \citep{nla}.
\begin{lemma}
Let $X^tX$ be the matrix given by
\begin{align*}
X^t X = \begin{bmatrix}
D_1 & B^t\\
B& D_2
\end{bmatrix},
\end{align*}
and $\tilde{R}$ be the matrix given by
\begin{align*}
\tilde{R} = \begin{bmatrix} \frac{\cos^2\theta}{\nu_2^2} I_{k_1} & 0 \\
0 & \frac{1}{\tan^2\theta} I_{k_2}
\end{bmatrix}.
\end{align*}
Then, 
\begin{align*}
\det{X^tX +\tilde{R}} = \det{\left(D_2+\frac{1}{\tan^2\theta}\right)}\, \det \left(D_1 + \frac{\cos^2\theta}{\nu_2^2} -B^t \left(D_2+\frac{1}{\tan^2\theta}\right)^{-1} B \right).
\end{align*}
\end{lemma}

In the following lemma, which follows immediately from elementary linear algebra operations, we show that $X^tX + \tilde{R}$ can be written as $D + UU^t$ where $U$ is a $k \times 2k_1$ matrix and $D$ is diagonal. This will be used to solve the linear system (\ref{340}) via the Sherman-Morrison-Woodbury formula.

\begin{lemma}
Suppose $X^tX$ and $\tilde{R}$ are the same matrices as defined previously. Then
\begin{align*}
X^t X +\tilde{R} = \begin{bmatrix} D_1+\frac{\cos^2\theta}{\nu_2^2} & 0 \\ 0 &D_2+\frac{1}{\tan^2\theta} \end{bmatrix} + \begin{bmatrix}I_{k_1} &0 \\ 0 & B \end{bmatrix} \begin{bmatrix}0&I_{k_1} \\ B & 0 \end{bmatrix}^t.
\end{align*}
Moreover,
\begin{align*}
\begin{bmatrix}0&I_{k_1} \\ B & 0 \end{bmatrix}^t \begin{bmatrix} \left(D_1+\frac{\cos^2\theta}{\nu_2^2} \right)^{-1} & 0 \\ 0 &\left(D_2+\frac{1}{\tan^2\theta} \right)^{-1} \end{bmatrix}\begin{bmatrix}I_{k_1} &0 \\ 0 & B \end{bmatrix}= \begin{bmatrix} 0&  B^t\left(D_2+\frac{1}{\tan^2\theta} \right)^{-1}B\\ \left(D_1+\frac{\cos^2\theta}{\nu_2^2} \right)^{-1} & 0\\ \end{bmatrix}
\end{align*}
\end{lemma}

The following theorem allows us to efficiently solve the linear system (\ref{340}) and follows immediately from the combination of the previous lemma with the Sherman-Morrison-Woodbury formula \citep{nla}.

\begin{theorem}
Let $X^tX$ and $\tilde{R}$ be as defined above. Let $W(\nu_3)$ be the $k_1 \times k_1$ matrix defined by
$$W(\nu_3) = B^t \left(D_2 + \frac{1}{\tan^2\theta}\right)^{-1} B.$$
Then
\begin{align*}
 = -\begin{bmatrix} P_1 & 0 \\ 0 &P_2 \end{bmatrix} \begin{bmatrix} I_{k_1} &0\\ 0 &B
 \end{bmatrix} \begin{bmatrix} I & W\\P_1 &I \end{bmatrix}^{-1}\begin{bmatrix}0&I_{k_1} \\ B & 0 \end{bmatrix}^t \begin{bmatrix} P_1 & 0 \\ 0 &P_2 \end{bmatrix}, 
\end{align*}
where $P_1 = (D_1 + \frac{\cos^2\theta}{\nu_2^2})^{-1},$ and $P_2 = (D_2 + \frac{1}{\tan^2\theta})^{-1}.$ Moreover, 
\begin{align*}
\begin{bmatrix} I & W\\P_1 &I \end{bmatrix}^{-1} = \begin{bmatrix} (I-P_1W)^{-1} & -W (I-P_1W)^{-t}\\ -P (I-P_1W)^{-1} & (I-P_1W)^{-t}\end{bmatrix}.
\end{align*}
\end{theorem}

Using the preceding theorems and the quadrature rule described in Section \ref{s60}, we compute posterior moments of the special case of $f$ described in this section.

\section{Mixed effects}\label{s102}
From a computational standpoint, the mixed effects model is nearly identical to the two-group normal-normal model, however they differ in one key respect. 
In the two-group normal-normal model (see (\ref{16})), the scale parameters $\sigma_1$ and 
$\sigma_2$ are treated as unknowns that are fit to the data and assigned priors---they're treated as modeled coefficients, also called ``random effects.'' 
In the mixed effects model, $\sigma_1$ is still a random effect, however instead
of treating $\sigma_2$ as a random effect, 
each regression coefficient in the second group of predictors, $\beta_{2,i}$ is given a
normal prior with fixed scale parameter. 
The corresponding Bayesian model is
\bb
\begin{split}
&y \sim \text{normal}(X_1\beta_1 + X_2\beta_2, \sigma_3) \\
&\beta_1 \sim \text{normal}(0, \sigma_1)\\
&\beta_{2,i} \sim \text{normal}(0, \sigma_{2,i}),
\end{split}
\ee
where $\sigma_{2,i}$ is the fixed scale parameter prior on each regression coefficient $\beta_{2,i}$ for $i=1,...,k_2$ where $\beta_2 \in \R^{k_2}$.
For the purposes of demonstrating the algorithm in this paper, we assign the priors
\begin{align*}
& \sigma_1 \sim \text{normal}^+(0, 1) \\
& \sigma_3 \sim \text{normal}^+(0, 1).
\end{align*}
The choice of priors on $\sigma_1$ and $\sigma_3$ is somewhat arbitrary. The algorithm described allows for general choices for these priors.  
The unnormalized posterior density for the mixed effects model, $f: \R^{k+2} \to \R$,
is given by
\begin{align*}
f(\beta,\sigma_1,\sigma_2) = \frac{e^{-\sigma_1^2/2-\sigma_2^2/2}}{\sigma_1^n \sigma_2^{k_1}} e^{-\frac{1}{2\sigma_1^2} \|X \beta - y\|^2}e^{-\frac{1}{2\sigma_2^2}\|\beta_1\|^2}e^{- \sum_{i=1}^{k_2} \frac{\beta_{2, i}^2}{2\sigma_{3, i}}},
\end{align*}
where $\sigma_3 \in \R^{k_2}$ is a vector of fixed scale parameter priors for the regression
coefficients $\beta_2 \in \R^{k_2}$. 
Now our density is over $k+2$ dimensions---the $k$ regression coefficients, $\beta = (\beta_1, \beta_2)$, and 
the two scale parameters, $\sigma_1$ and $\sigma_2$. 

With a change of variables we convert the posterior density, $f$,
to the posterior density of a two-group normal-normal model where one 
scale parameter is fixed. That is, we now convert $f(\beta, \sigma_1, \sigma_2)$
to $q(\beta, \sigma_1, \sigma_2, 1)$ where $q$ is the posterior density of 
a two-group normal-normal model (see (\ref{10})). 

We first scale the last $k_2$ columns of $X$ by $\sigma_3^2$ and define
$\hat{X}$ to be resulting matrix:
\begin{align*}
\hat{X}_{i, j} = X_{i, k_1 + j} \sigma_{3, j}^2
\end{align*}
for $i=1,...,n$ and for $j = 1,...,k_2$.
We define $\hat{\beta_2}$ to be the vector
\begin{align*}
\hat{\beta}_{2, i} = \beta_{2, i} \sigma_i^2,
\end{align*}
and define $\hat{f}$ by the unnormalized density
\begin{align*}
\hat{f}(\hat{\beta}, \sigma_1, \sigma_2) = \frac{e^{-\sigma_1^2/2-\sigma_2^2/2}}{\sigma_1^n \sigma_2^{k_1}} e^{-\frac{1}{2\sigma_1^2} \|\hat{X} \hat{\beta} - y\|^2}e^{-\frac{1}{2\sigma_2^2}\|\beta_1\|^2}e^{-\frac{1}{2} \| \hat{\beta}_2 \|^2}.
\end{align*}
It follows that posterior means and standard deviations of $\beta_2$ become
\begin{align*}
E_f[\beta_{2, i}] = \sigma_i^2 E_{\hat{f}}[\hat{\beta}_{2, i}], \qquad E_f[(\beta_{2, i} - E_f[\beta_{2, i}])^2] = \sigma_i E_{\hat{f}}[(\hat{\beta_{2, i}} - E_{\hat{f}}[\hat{\beta}])^2],
\end{align*}
and all other posterior first and second moments are unchanged under density $\hat{f}$.
We've now reduced the problem of finding moments of $f$ to
finding moments of $\hat{f}$, which is equal to $q(\hat{\beta}, \sigma_1, \sigma_2, 1)$
where $q$ is defined in (\ref{10}).

At this point, we rely on the analysis and numerical tools of the 
two-group normal-normal model for evaluating posterior moments. The only difference between 
evaluation of moments of the two-group model is that in the two group model
the marginal density is a $3$-dimensional density over $(\sigma_1, \sigma_2, \sigma_3)$
whereas in the mixed effects model, the marginal density is over two dimensions, 
$(\sigma_1, \sigma_2)$ and $\sigma_3 = 1$ is fixed. 
As a result, we perform the change of variables 
from $(\sigma_1, \sigma_2, 1)$ to polar 
coordinates 
\begin{align*}
& \sigma_1 = \rho\cos(\phi) \\
& \sigma_2 = \rho \sin(\phi) \cos(\theta) \\
& 1 = \rho \sin(\phi) \sin(\theta),
\end{align*}
or equivalently
\begin{align*}
& \rho = \sqrt{\sigma_1^2 + \sigma_2^2 + 1} \\ 
& \theta = \text{atan}(1 / \sigma_2) \\
& \phi = \text{acos}\bigg(\frac{\sigma_1}{\sqrt{\sigma_1^2 + \sigma_2^2 + 1}}\bigg).
\end{align*}
The differentials become
\begin{align*}
d\sigma_1 d\sigma_2 = | \, \gamma \, | ^{-1} \, d\theta d\phi,
\end{align*}
where
\begin{align*}
\gamma  \, = \, \frac{-1}{1 + \sigma_2^2} 
\bigg( \frac{1}{\sqrt{\alpha}} - \frac{\sigma_1^2}{\alpha^{3/2}} \bigg)\
\bigg( 1 - \frac{\sigma_1^2}{\alpha} \bigg)^{-1/2}
\end{align*}
and
\begin{align*}
\alpha = \sigma_1^2 + \sigma_2^2 + 1.
\end{align*}
We can now evaluate the posterior of $\hat{f}$ with Algorithm \ref{a10}, where
the integral with respect to $\rho$ is replaced with $\rho = \sqrt{\sigma_1^2 + \sigma_2^2 + 1}$.

\section{A simple example: Hierarchical linear model}\label{sec:rats}

We demonstrate the algorithm on a hierarchical linear model describing the growth of a group of young rats over a period of several weeks; this is a small example that has been used in the statistical literature \citep{gelfand:1990}.
In the experiment, the weight of each rat is measured at regular time intervals. 
Regression coefficients are computed for each rat; that is, for the $j^\mathrm{th}$ rat, we estimate an intercept $\alpha_j$ and a linear coefficient $\beta_j$.
We assign a normal prior on both parameters and estimate the prior scale.
The full model is as follows:
\begin{equation}\label{717}
\begin{split}
  y_i & \sim \mathrm{normal}(X^1_i \alpha + X^2_i \beta, \sigma_1) \\
  \alpha_j & \sim \mathrm{normal}(0, \sigma_2) \\
  \beta_j & \sim  \mathrm{normal}(0, \sigma_3) \\
  \sigma_k & \sim \mathrm{normal}^+(0, 10) \mbox{ for } k=1,2,3,
\end{split}
\end{equation}
where $X_1$ is an indicator matrix indicating to which rat each observation (weighing) corresponds; $X_2$ is that same indicator matrix multiplied by $w - \bar w$, where $w$ is the observation week and $\bar w$ the mean observation week.
In other words, we have an intercept and a slope parameter for each rat.
The data is centered at 0 and the priors on the scale parameters are  weakly informative.

We demonstrate the efficiency of the two-group normal-normal Algorithm \ref{a10} on evaluating posterior means and standard deviations of the rats model on randomly generated data. We assume an experiment with $100$ rats and $20$ weighing times and randomly generated data for each weighing. As a result, matrix $X_1$ and $X_2$ of model (\ref{717}) are $2000 \times 100$ matrices. 

Because the data size is relatively small and the data matrices have a friendly structure, running MCMC with Stan (4 chains in parallel, each with 1,000 warmup iterations and 1,000 sampling iterations) only takes 13.9s.
Our algorithm takes 1.6s and achieves significantly smaller errors than MCMC estimates.
Figure~\ref{fig:rats_err} shows the error of the MCMC estimates as a function of time and the accuracy achieved by our algorithm. 
For problems with larger data, the difference in time scale becomes important.

\begin{figure}
    \centering
    \includegraphics[width=\textwidth]{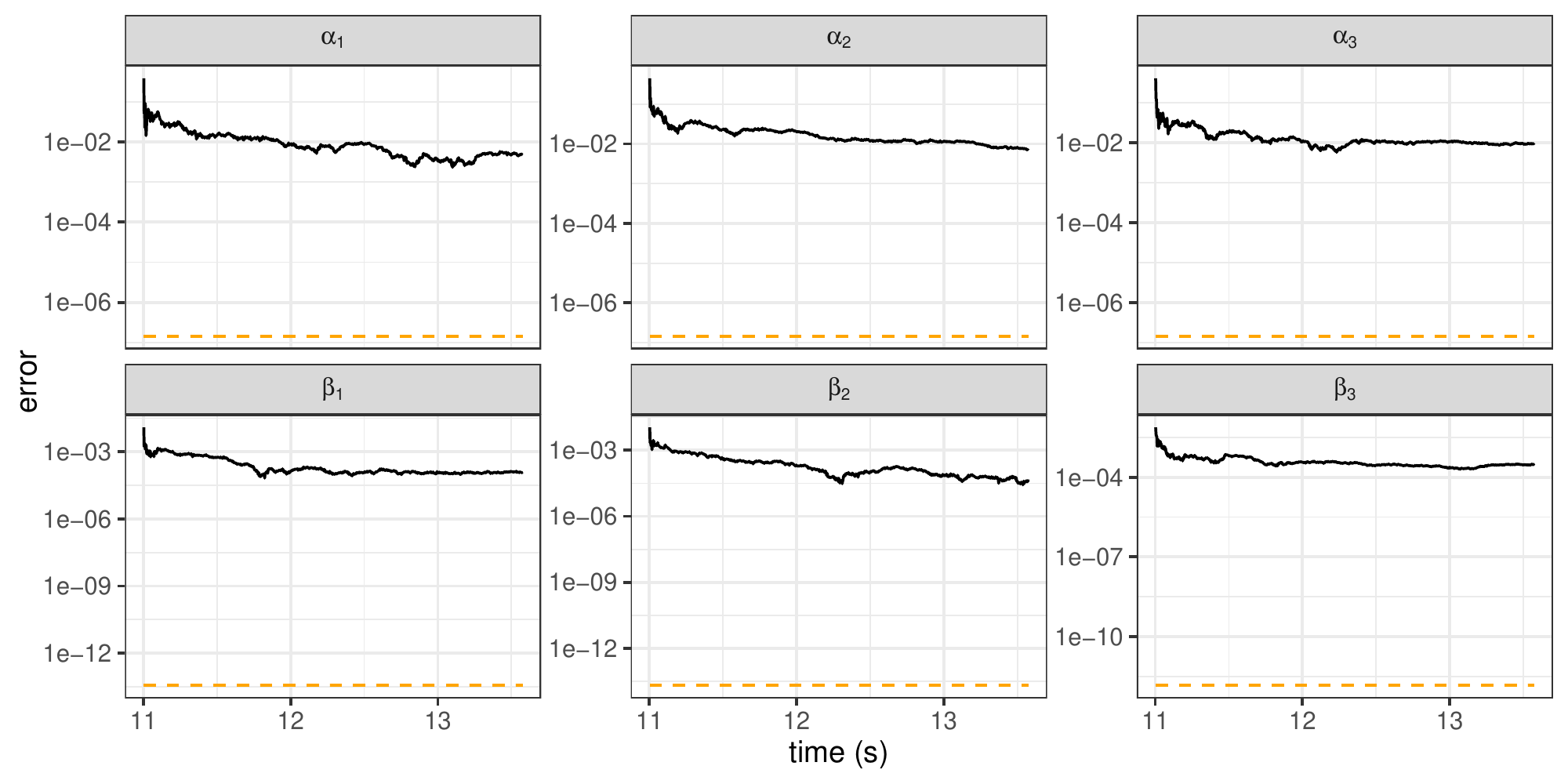}
    \caption{\em Error of MCMC estimates via Stan as a function of run time. The horizontal orange line is the error of Algorithm \ref{a10}.}
    \label{fig:rats_err}
\end{figure}

\section{Application: COVID-19 symptom survey} \label{sec:covid}

As of the writing of this article, the coronavirus pandemic is still raging in
many countries and stressing healthcare systems around the world.
A challenge at the start of the pandemic was tracking its spread, 
especially in locations where reliable testing was not widely available. Having accurate 
estimates of infection rates across geographical regions can be extremely
helpful. For example, reliable estimates allow hospital systems to allocate resources 
efficiently, they can alert residents of the need to take extra precaution 
in their daily routines, and they can facilitate better policy from local governments.
In order to get improved estimates of infection rates in the absence of widespread testing, initiatives were deployed in early 2020 in several countries that
allowed individuals to report symptoms via publicly available surveys 
\citep[e.g.,][]{Segal:2020}.

One country where these surveys provided valuable information was 
Israel  \citep{Rossman:2020}, where demographic and health data was 
provided by tens of thousands of respondents across the country. 
The large amount of data collected from survey respondents provided 
data scientists and policy-makers with a great resource, however at the 
same time, large amounts of data turns the computational aspect of 
statistical modeling into a substantial challenge. 

In this section, we present an exploratory model used to analyze data 
from the COVID-19 survey conducted in Israel \citep{Rossman:2020}. Using straightforward
MCMC with Stan \citep{carpenter1} was inconvenient; using the full
data set resulted in runtimes of several hours. Using the algorithms
of this paper, we were able to evaluate posterior moments in seconds without 
loss of accuracy. 

\subsection*{Multilevel regression and poststratification procedure}

The respondents are anonymous, but several of their features are recorded, including their age and the city in which they live.
We can use the data to identify regions in which the average symptom score seems unusually high.


A first exploratory model uses an intercept, age group, and population density in the respondent's city, as covariates, $X$, and an indicator matrix $Z$ for city:
$$y_i \sim \mathrm{normal}(X\beta + Z u, \sigma_1),$$
with a hierarchical prior on the city parameters,
$$u_j\sim \mathrm{normal}(0,\sigma_2),$$
and weakly informative priors on the other coefficients,
$$\beta_j\sim\mathrm{normal}(0,1).$$
This unit prior is weakly informative if the outcome $y_i$ has been standardized and the continuous predictors (in this case, population density) has also been standardized to be on unit scale.

In addition, we put weakly informative half-normal priors (standard normal distributions restricted to the non-negative reals) on the hyperparameters $\sigma_1$ and $\sigma_2$:
\begin{eqnarray*}
  \sigma_1 & \sim & \mathrm{normal}^+(0, 1) \\
  \sigma_2 & \sim & \mathrm{normal}^+(0, 1).
\end{eqnarray*}
This corresponds to a two-group normal-normal model with an additional covariate.
In cities where $u$ cannot be well estimated due to a low response rate, we can rely on the rest of the model, that is a regression model based on age and population density.

Only a fraction of the population responds to the survey, which raises questions about biases.
This is notably a concern because different age groups behave differently: not only do their chances of contracting and spreading the disease vary, their susceptibility to the disease also changes.
In multilevel regression and poststratification (MRP), we adjust for these biases by using  estimates of the proportion of people in each city that belong to each age group. 
For this model, the proportions are estimated using census data.
This leads to a corrected estimate for the expected symptom score of an individual in city $i$:
\begin{equation*}
    \tilde u_i = u_i + \beta_0 + \beta_{\mathrm{density}} d_{i} + \sum_{j = 1}^n a^i_j \beta_{\mathrm{age}, j},
\end{equation*}
where $\beta_0$ is the intercept, $\beta_{\mathrm{density}, i}$ is the regression coefficient of the population density covariate, $d_i$ denotes the density of city $i$, 
$a^i_j$ is the proportion of individuals in the $j^\mathrm{th}$ age group in the $i^\mathrm{th}$ city, and $\beta_{\mathrm{age}, j}$ is the regression coefficient of age group $j$. 

Using the means and covariances of $u, \beta_0, \beta_\mathrm{density}$, and $\beta_{\mathrm{age}}$ we compute the posterior mean and variance for $\tilde u$, per the following formulas.
Given a linear combination of random variables, $Y = \sum_i \delta_i Z_i$, we have
\begin{equation*}
    E Y = \sum_i \delta_i E Z_i,
\end{equation*}
and
\begin{equation*}
    \mathrm{Var} Y = \sum_i \delta_i^2 \mathrm{Var} Z_i + 2 \sum_{i < j} \delta_i \delta_j \mathrm{Cov}(Z_i, Z_j).
\end{equation*}
Moreover, variance and thence standard deviations of $\tilde u$ can be computed, provided we also evaluate the relevant posterior covariances.

\subsection*{Comparison of our algorithm to MCMC}

We analyze the data collected over the two weeks between April $15^\mathrm{th}$ and $30^\mathrm{th}$ 2020, across 351 cities.
These are cities for which we know, through census data, the population density and the age distribution.
The total number of responses is 135,501.

Our proposed algorithm returns the posterior mean and standard deviation for all variables of interest and takes $\sim$7s to run.

We next fit the model in Stan using the default dynamic HMC sampler.
After warming up the sampler for 500 iterations, we compute another 500 draws, using 4 chains computed in parallel, for a total of 2,000 sampling iterations.
The wall time for this procedure is $\sim$12,000s ($>$3 hours).
For each city, we computed the Monte Carlo mean.
Figure~\ref{fig:covid_mrp} plots the posterior mean and standard deviation of $\tilde u$ for all cities, computed by both methods.
Figure~\ref{fig:covid_error} shows the difference between our algorithm and the Monte Carlo estimate, as a function of computation time.
While it takes on the order of hours to get accurate results with MCMC, our algorithm achieves better results within seconds.

\begin{figure}
 \centering
  \resizebox{\linewidth}{!}{
    \includegraphics{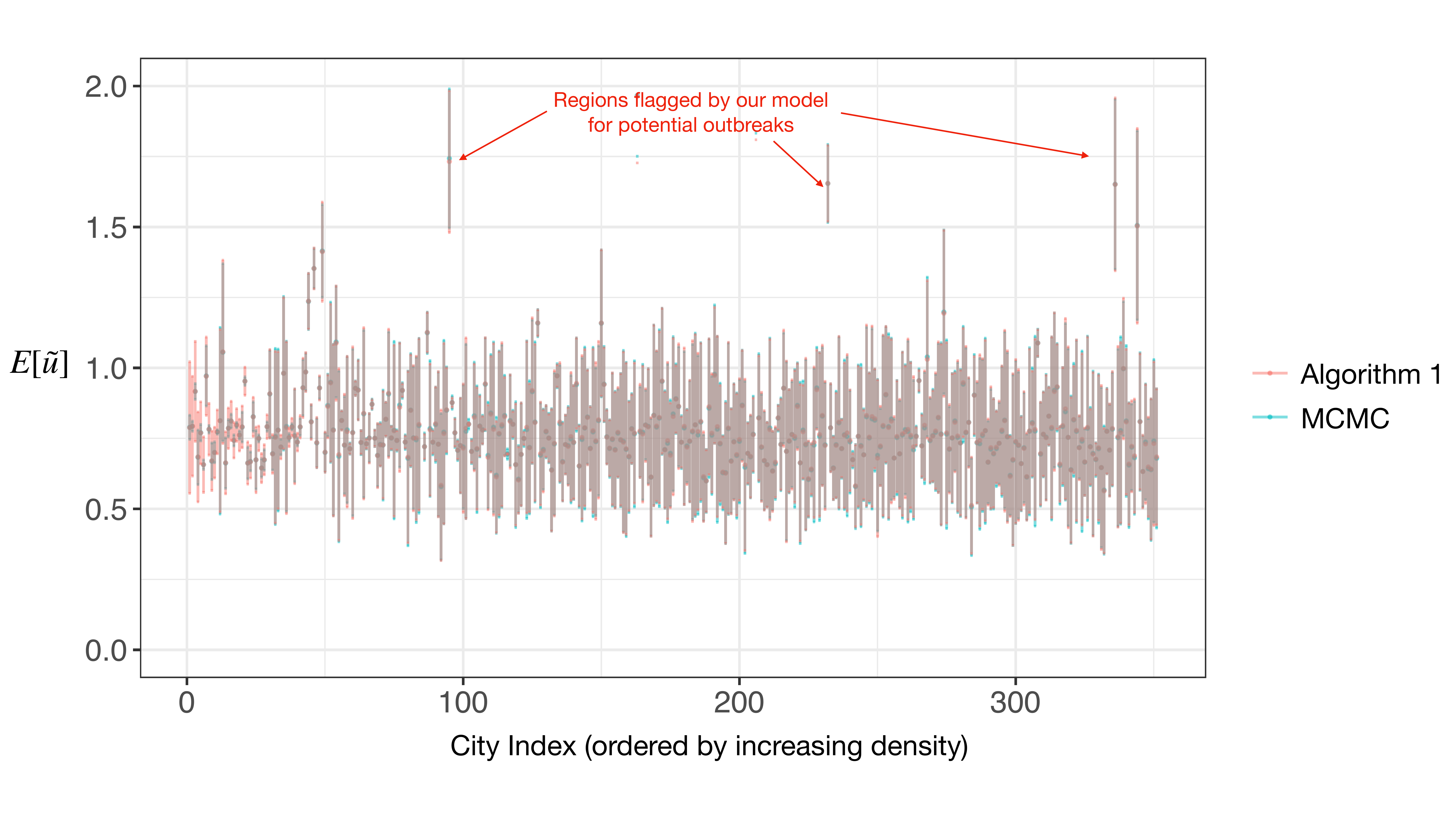}
    }
    \caption{\em Posterior mean and standard deviation for $\tilde u$ computed using Algorithm 1 and MCMC. The points represent the estimated mean and the ``error bars'' span two standard deviations.}
    \label{fig:covid_mrp}
    
\end{figure}

\begin{figure}
    \centering
    \includegraphics[width=6in]{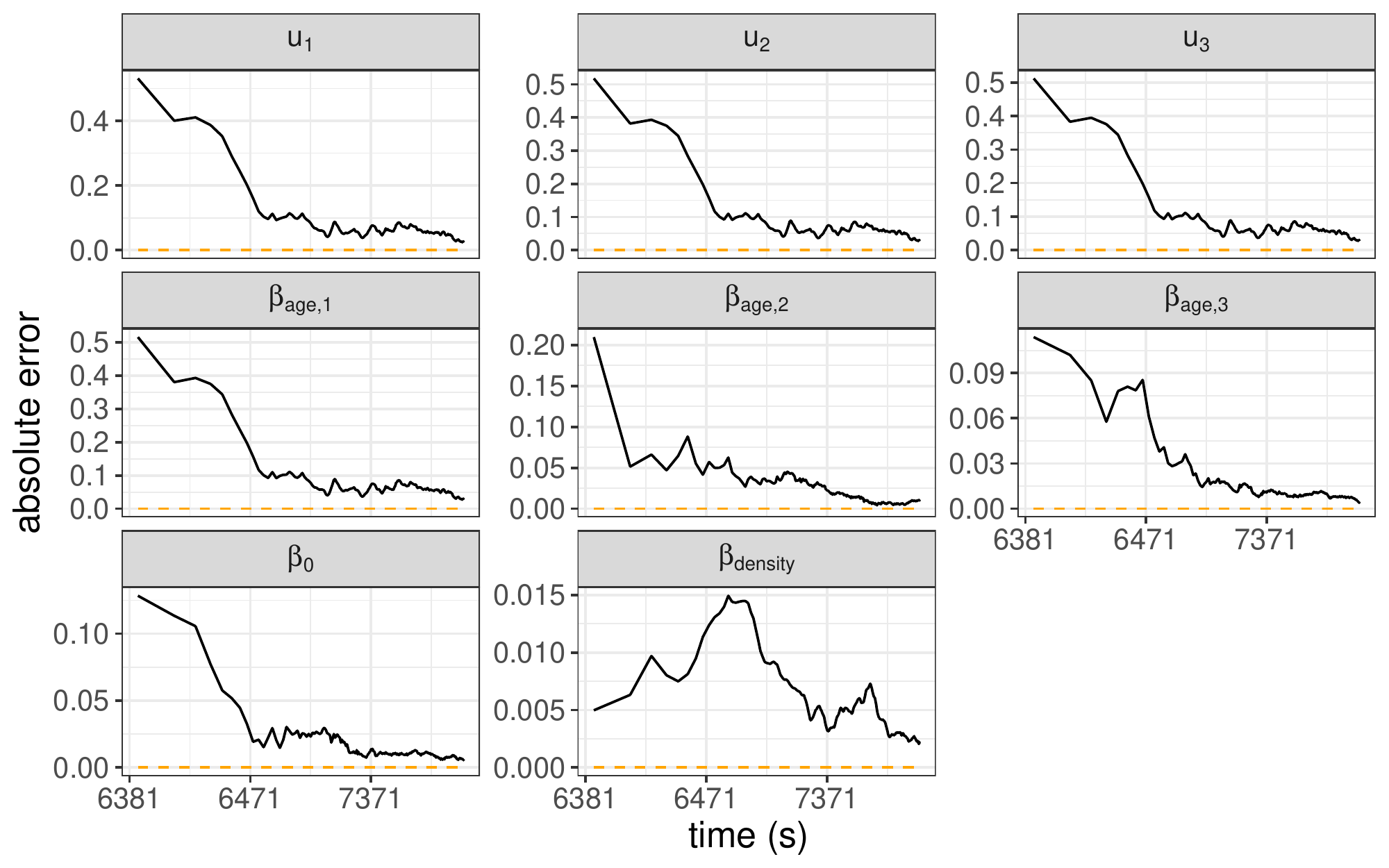}
    \caption{\em Error of MCMC estimates via Stan as a function of run time.
As a benchmark, we use the estimate returned by Algorithm \ref{a10}.
    The horizontal orange line is the error of our method, which took $7$ seconds to run. MCMC in Stan required over $6,000$ seconds of warmup before error can be measured. Total run time for Stan was $\sim12,000$ seconds.}
    \label{fig:covid_error}
\end{figure}

\begin{table}
    \centering
    \begin{tabular}{c | c | c |}
        Regression & MCMC Accuracy & Algorithm \ref{a10} Accuracy \\
         coefficient & ($\sim$12,000 s) & (7 s) \\\hline
        $\beta_0$ & 3e-02 & 1e-05 \\
        $u_{1}$ & 1e-02 & 4e-04 \\
        $u_{2}$ & 1e-02 & 3e-04 \\
        $u_{3}$ & 2e-02 & 6e-05 \\
        $\beta_{\mathrm{age}, 1}$ & 3e-02 & 7e-06\\
        $\beta_{\mathrm{age}, 2}$ & 3e-02 & 8e-06  \\
        $\beta_{\mathrm{age}, 3}$ & 3e-02 & 9e-08 \\
        $\beta_{\mathrm{density}}$ & 2e-03 & 2e-05 \\
    \end{tabular}
    \caption{\em Accuracy of the approximation of posterior means for several regression coefficients with both MCMC and Algorithm \ref{a10}. For MCMC, the total time for the approximation was $\sim$12,000s. Total time for Algorithm \ref{a10} was 7s.}
    \label{tab:my_label}
\end{table}

\subsection*{Limitations of the model and our numerical method}

We believe the presented model offers an improvement on the analysis conducted on the survey data \citep{Rossman:2020}, because (i) it uses full Bayesian inference to quantify uncertainty and (ii) it corrects sampling biases using a poststratification step.
A more careful quantification of uncertainty would use posterior intervals, rather than posterior variance. Such an interval can be estimated using MCMC draws.
Extending our numerical scheme into a sampling scheme to estimate such intervals is a direction we are actively pursuing.

For the model of this paper we only used a fraction of the available covariates,
that is, the data collected in survey responses. As a result, the model 
can be extended to include more than two groups.
Estimates tend to be noisy because the studied covariates can be strongly 
correlated with the outcome. For example, age is correlated with intensity 
of symptoms. The marginal correlation however is weak.
This, and other considerations, suggest that it might be beneficial from a modeling 
standpoint to build a more sophisticated model, which might be outside of the 
scope of application of the methods of this paper.
Nevertheless, the model considered here is an important step in the development of a better model.

\section{Application: Public opinion on abortion policies}\label{sec:abort}

We next apply our method to a hierarchical linear regression used to model attitudes on abortion policies as they vary across states, ethnicity, age groups, and education levels.
Modeling this heterogeneity requires partitioning an initially large data set into small groups.
Furthermore, we must address biases that can arise in our survey and correct them using more comprehensive surveys, such as census data.
As in Section~\ref{sec:covid}, we use MRP to do inference for small slices of big data and correct biases in our survey.

We analyze data from the 2018 Cooperative Congressional Election Study (CCES) using, as in the case study of \citep{Lopez-Martin:2020}, a random subset of 5,000 respondents.
Respondents express support or opposition on six abortion policies, for example ``Ban abortion after the 20$^\text{th}$ week of pregnancy'' or ``Allow employers to decline coverage of abortion in insurance plan.''
These policies are intended to restrict access to abortion.
Each respondent is given a support score, $y$, ranging from 0 to 6, indicating the number of supported policies.

We use a normal likelihood with the following covariates, recorded for each respondent: state, ethnicity, age group, education level, and sex.
We use the proportion of votes for the Republican party in the state in 2016 as an additional predictor, denoted as \texttt{repvote}.
The model also admits an intercept term.
The statistical formulation of the model is the following:
\begin{eqnarray*}
    y_i & \sim & \text{normal}(\beta_0 +
    X^\text{state}_i \beta_\text{state} +
    X^\text{ethnicity}_i \beta_\text{ethnicity} +
    X^\text{age} \beta_\text{age}\\
    & & \ \ \ \ \ \ \ \ \ \ +
    X^\text{sex}  \beta_\text{sex} + X^\text{education} \beta_\text{education} + X^\text{repvote} \beta_\text{repvote}, \,\sigma_1) 
\end{eqnarray*}
The difficult parameters to estimate here are the state coefficients, to which we give $\mbox{normal}(0,\sigma_2)$ priors.  Because the model includes \texttt{repvote}, the partial pooling is done toward the prediction of the state based on its previous vote, not toward the national mean.

Table~\ref{tab:abortion} summarizes the performance of the algorithm on this model.
The posterior mean and standard deviation of the MRP estimates for each state can be computed as in Section~\ref{sec:covid} and are plotted in Figure~\ref{fig:abort}.

\begin{table}[!ht]
    \centering
    \begin{tabular}{l l l l c}
         $n$ & $k_1$ & $k_2$ & max error & total time (s)  \\
         \hline
         5000 & 50 & 19 & $1.2 \times 10^{-8}$ & 0.05
    \end{tabular}
    \caption{\em Computation time and accuracy of Algorithm \ref{a10} applied to a model of support/opposition for abortion policies. The column ``max error" shows the maximum error of posterior means and standard deviations of regression coefficients and scale parameters.}
    \label{tab:abortion}
\end{table}

\begin{figure}
    \centering
    \includegraphics[width = 6in]{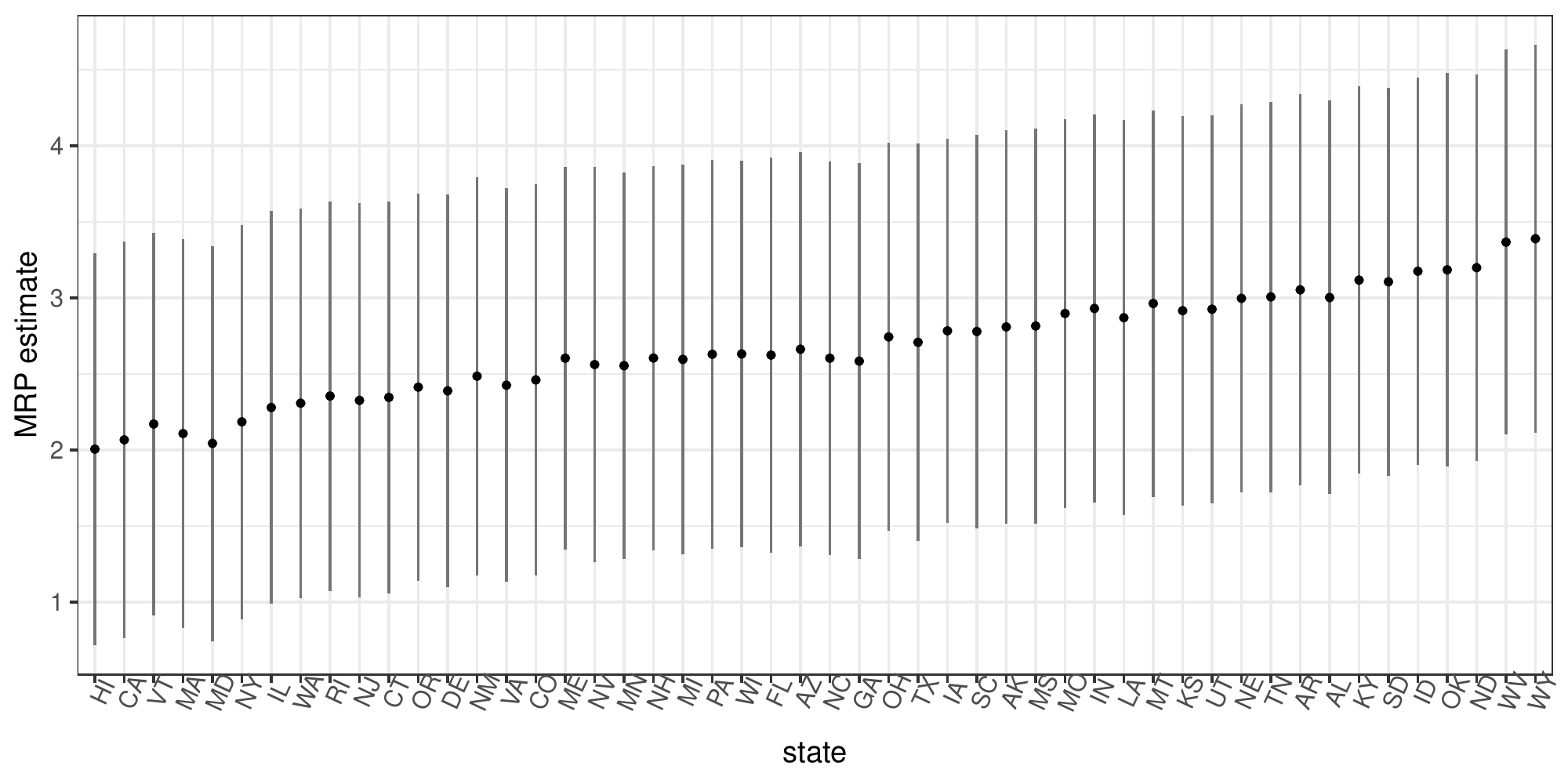}
    \caption{\em MRP estimate of the expected level of support for anti-abortion policies in each state. The point represents the posterior mean, and the bars span two posterior standard deviations. The states are ordered based on Republican vote share in the 2016 presidential election.}
    \label{fig:abort}
\end{figure}

Figure~\ref{fig:abort} shows that the expected support score increases with the level of support for the Republican party, bearing some fluctuations.
The large posterior standard deviations indicate there is quite a bit of heterogeneity within each state.
For further insight, we may examine how groups other than states, e.g. ethnic groups, age groups, etc. behave.

The present model has certain limitations.
First, one could consider interaction terms.
This seems sensible since, for instance, white males with no college education likely behave differently than white males with a college degree.
The numerical method presented in this paper can handle interaction terms.
Computing the posterior standard deviation of the MRP estimate however requires some data wrangling.
We plan to create an R package with routines that seamlessly implement these MRP calculations, making it straightforward for modelers to experiment with different covariates and interaction terms.

There is also interest in nonlinear models with non-normal likelihoods.
\citet{Lopez-Martin:2020} consider an item-response or ideal-point logistic regression.  This sort of model can better capture certain characteristics of the data, such as dependence among different survey responses.
For such models, we cannot use the proposed integration scheme.
This presents us with a tradeoff: the proposed algorithm takes a fraction of a second to run, while fitting the ideal point model with Stan's MCMC takes hundreds of seconds.
The difference is more severe if, rather than fitting a subset of 5,000 respondents, we use all 60,000 respondents in the survey.
The modeler then needs to assess how useful it is to use a non-normal likelihood.
Even then, the normal likelihood model can be a fast way to do model exploration, by for example examining various covariates and interaction terms.

\section{Conclusions and generalizations}\label{sec:conclusion}
In this paper we describe a class of fast algorithms for evaluating the posterior moments of two Bayesian linear regression models:

\begin{enumerate}
\item Two-group normal-normal: The two-group normal-normal model is used to model a continuous outcome with two groups of parameters:
\bb
\begin{split}
&y \sim \text{normal}(X_1\beta_1 + X_2\beta_2, \sigma_3) \\
&\beta_1 \sim \text{normal}(0, \sigma_1)\\
&\beta_2 \sim \text{normal}(0, \sigma_2)
\end{split}
\ee
where $X_1$ is a $n \times k_1$ matrix of predictors, $\beta_1 \in \R^{k_1}$ is a vector
of regression coefficients, $X_2$ is a $n \times 
k_2$ matrix of predictors,  and $\beta_2 \in \R^{k_2}$ is a vector of regression coefficients.

\item Mixed effects model: 
The mixed-effects model is a slight variant of the two-group normal-normal model. In the mixed-effects model we model the scale parameter on one group of coefficients and  assign fixed scale parameters to the priors on all other coefficients:
\bb
\begin{split}
&y \sim \text{normal}(X_1\beta_1 + X_2\beta_2, \sigma_3) \\
&\beta_1 \sim \text{normal}(0, \sigma_1)\\
&\beta_{2,i} \sim \text{normal}(0, \sigma_{2,i})
\end{split}
\ee
where $\sigma_{2,i}$ is the fixed scale parameter prior on each regression coefficient $\beta_{2,i}$ for $i=1,...,k_2$ where $\beta_2 \in \R^{k_2}$.
\end{enumerate}

The algorithms of this paper allow for assigning a general choice of priors on the scale parameters. We demonstrated the performance of our algorithm for posterior inference on two applications. In Section \ref{sec:covid} we used COVID-19 symptom survey data to model geographic and age effects. We also used the mixed-effects model with public opinion survey data to estimate geographic and demographic impacts on attitudes towards abortion.  These are both existing applications that have been fit with MCMC; by allowing these models to be fit much faster, our algorithm can facilitate a workflow in which users can fit and explore many more models in real time.

The algorithms of this paper provide substantial improvements over standard MCMC methods in both computation time and accuracy in approximating posterior moments. 
These improvements rely on analytically integrating the regression coefficients, which make up the bulk of the posterior dimensions, and then numerically integrating the remaining low-dimensional density with Gaussian quadrature.

Many of the techniques and analysis used in this paper generalize to multilevel and multigroup models with more than two-groups. For an $m$ group model, the numerical integration of our algorithm is computed over a $m+1$ dimensional density. For models with large $m$ (large number of groups) the analytic marginalization of this paper can still be applied, however, integration via a tensor product of Gaussian nodes will not be feasible. On the other hand, using MCMC or 
other integration schemes can be used on the $m+1$ dimensional marginal density. 

Bayesian models with more than two groups and non-Gaussian likelihoods are directions  of future research.

\section{Acknowledgements}
The authors are grateful to Hagai Rossman and Ayya Keshet for useful discussions and their contribution to the COVID-19 model.

\appendix
\section{Integral with respect to \texorpdfstring{$\rho$}{ρ}}\label{app1}
In this section we describe analytical properties of the integrand of the inner integral of (\ref{10}) that are used in the evaluation of the integral. 

Let $\psi : \bR^{+3} \rightarrow \bR$ be defined by the formula
\begin{equation}\label{330}
\psi(\rho, c, n) = 
\frac{e^{-\frac{\rho^2}{2}-\frac{c}{2\rho^2}}}{\rho^{n-2}}.
\end{equation}
We seek, for fixed $c$ and $n$, the value for $\rho$ that maximizes
$\psi(\rho, c, n)$. 
We observe that 
\bb
\begin{split}
\frac{\partial \psi}{\partial \rho} 
&= e^{-\frac{\rho^2}{2}-\frac{c}{2\rho^2}}
\left(\rho^{1-n}(2-n) + \left(\frac{c}{\rho^3}-\rho\right)\rho^{2-n}\right) \\
&= e^{-\frac{\rho^2}{2}-\frac{c}{2\rho^2}}
\rho^{1-n} \left(2-n + \frac{c}{\rho^2}-\rho^2\right)
\end{split}
\ee
Setting
\bb
2 -n + \frac{c}{\rho^2} - \rho^2 = 0
\ee
and rearranging terms, we have
\bb
\rho^4 + \rho^2(n-2) -c = 0.
\ee
Then setting
\begin{equation}
\rho_{max} = \frac{1}{\sqrt{2}}\left(\sqrt{4c+(n-2)^2}-n+2\right)^{1/2}
\end{equation}
we observe
\bb
\frac{\partial \psi}{\partial \rho} (\rho_{max}) = 0.
\ee
That is, for fixed $c, n$, we have $\psi$ achieves its maximum at 
\begin{equation}\label{320}
\rho = \frac{1}{\sqrt{2}}\left(\sqrt{4c+(n-2)^2}-n+2\right)^{1/2}.
\end{equation}
Furthermore, 
\bb
\frac{\partial}{\partial \rho}\log(\psi(\rho, c, n))
= -\rho + \frac{c}{\rho^3} - \frac{n-2}{\rho}
\ee
and
\bb
\frac{\partial^2}{\partial \rho^2}\log(\psi(\rho, c, n))
= -1 - \frac{3c}{\rho^4} + \frac{n-2}{\rho^2}.
\ee

\bibliography{refs}
\end{document}